\documentclass[11pt]{article}
\usepackage{kantlipsum}
\usepackage{amssymb, latexsym, mathtools, amsthm}
\begingroup
    \makeatletter
    \@for\theoremstyle:=definition,remark,plain\do{%
        \expandafter\g@addto@macro\csname th@\theoremstyle\endcsname{%
            \addtolength\thm@preskip\parskip
            }%
        }
\endgroup
\usepackage{enumerate}
\usepackage{pgf,tikz}
\usepackage{pdfsync}
\usepackage{txfonts}
\usepackage[T1]{fontenc}
\usepackage{booktabs}
 \usepackage{graphicx}
\definecolor{dnrbl}{rgb}{0,0,0.3}
\definecolor{dnrgr}{rgb}{0,0.3,0}
\definecolor{dnrre}{rgb}{0.5,0,0}
\usepackage[colorlinks=true, citecolor=dnrgr, linkcolor=dnrre, urlcolor=dnrre] {hyperref}
\usepackage{xcolor}
\usepackage{parskip}
\usepackage{tabularx, colortbl}

\usepackage{geometry}
 \usepackage[scriptsize, up]{caption}
\usepackage{pgf,tikz}
\usetikzlibrary{decorations, decorations.pathmorphing}
\usetikzlibrary {shapes}

\theoremstyle{plain}
\newtheorem{thm}{Theorem}[section]

\newtheorem{lem}[thm]{Lemma}
\newtheorem{coro}[thm]{Corollary}
\newtheorem{defi}[thm]{Definition}
\theoremstyle{definition}
\newtheorem{rem}[thm]{Remark}
\numberwithin{equation}{subsection}
\makeatletter

\let\c@table\c@figure
\makeatother

\newcommand{\Nat}{\mathbb{N}}

\newcommand{\restr}{\upharpoonright}  
\newcommand{\de}{\downarrow} 

\newcommand{\HNF}{Hirschfeldt, Nies, and Stephan\ }

\newcommand{\BFGM}{Becher, Figueira, Grigorieff, and Miller\ }
\newcommand{\KS}{Ku{\v{c}}era and Slaman\ }
\newcommand{\CHKW}{Calude, Hertling, Khoussainov, and Wang\ }

\newcommand{\ml}{Martin-L\"{o}f }
\newcommand{\pz}{$\Pi^0_1$\ }

\newcommand{\pzt}{$\Pi^0_2$\ }

\newcommand{\szt}{$\Sigma^0_2$\ }

\newcommand{\eg}{e.g.\ }
\newcommand{\ie}{i.e.\ }
\newcommand{\ce}{c.e.\ }

\newcommand{\dce}{d.c.e.\ }
\newcommand{\lce}{left-c.e.\ }

\newcommand{\rce}{right-c.e.\ }

\newcommand{\pf}{prefix-free }

\renewenvironment{abstract}
 { \normalsize
  \list{}{
    \setlength{\leftmargin}{.0cm}%
    \setlength{\rightmargin}{\leftmargin}%
    }%
  \item {\bf \abstractname.} \relax}
 {\endlist}

\newcommand{\inctot}{\mathtt{INCTOT}}

\DeclarePairedDelimiter{\dbra}{\llbracket}{\rrbracket}
\newcommand{\tot}{\mathtt{TOT}}
\newcommand{\rcep}{$\mathbf{0}'$-right-c.e.\ }
\newcommand{\lcep}{$\mathbf{0}'$-left-c.e.\ }

 \makeatletter
\newtheorem*{rep@theorem}{\rep@title}
\newcommand{\newreptheorem}[2]{%
\newenvironment{rep#1}[1]{%
 \def\rep@title{#2 \ref{##1}}%
 \begin{rep@theorem}}%
 {\end{rep@theorem}}}
\makeatother
 \newreptheorem{thm}{Theorem}
\newreptheorem{lem}{Lemma}
\usepackage{titlesec}
 \titleformat{\paragraph}
{\normalfont\normalsize\em\bfseries}{\theparagraph}{1em}{$\blacktriangleright$\hspace{0.2cm}}
\titlespacing*{\paragraph}
{0pt}{2.25ex plus 1ex minus .2ex}{0.5ex plus .2ex}

\title{The probability of a computable output from a random oracle
\thanks{Barmpalias was supported by the 
1000 Young Talents Plan from the Chinese Government, grant no.\ D1101130.
Additional support was received by
the Chinese Academy of Sciences (CAS) and the Institute of Software of the CAS.
Cenzer was partially supported by the 
U.S.\ National Science Foundation SEALS grant NSF DMS-1362273. Porter was supported
by the National Security Agency Mathematical Sciences Program grant H98230-I6-I-D310 as part of the Young Investigator's Program. The authors thank the referees for their input, which improved the presentation of
this article.}}

\author{George Barmpalias  \and Douglas Cenzer \and Christopher P.~Porter}
\date{\today}
\begin{document}
\maketitle
\begin{abstract}
Consider a universal oracle Turing machine that 
prints a finite or an infinite binary sequence, 
based on the answers to the binary queries that it makes
during the computation. We study the probability that this output is
infinite and computable, 
when the machine is given a random (in the probabilistic sense) 
stream of bits as the answers to its queries during an infinitary computation.
Surprisingly, we find that these  
probabilities are the entire class of real numbers in $(0,1)$ that can be
written as the difference of two halting probabilities relative to the halting problem. In particular,
there are universal Turing machines which produce a computable infinite output with probability exactly
1/2. Our results contrast a large array of facts 
(the most well-known being the randomness of Chaitin's halting probability) 
which witness maximal initial segment complexity of probabilities associated with universal machines.
Our proof uses recent advances in algorithmic randomness.
\end{abstract}
\vspace*{\fill}
\noindent{\bf George Barmpalias}\\[0.5em]
\noindent
State Key Lab of Computer Science, 
Institute of Software, Chinese Academy of Sciences, Beijing, China.
School of Mathematics, Statistics and Operations Research,
Victoria University of Wellington, New Zealand.\\[0.2em] \textit{E-mail:} \texttt{barmpalias@gmail.com}.
\textit{Web:} \texttt{\href{http://barmpalias.net}{http://barmpalias.net}}\par\medskip
\noindent{\bf Douglas Cenzer}\\[0.5em]
\noindent Department of Mathematics
University of Florida, Gainesville, FL 32611\\[0.2em]
\textit{E-mail:} \texttt{cenzer@math.ufl.edu.}
\textit{Web:} \texttt{\href{http://people.clas.ufl.edu/cenzer}{http://people.clas.ufl.edu/cenzer}}\par\medskip
\noindent{\bf Christopher P.~Porter}\\[0.5em]
\noindent Department of Mathematics and Computer Science,
Drake University,
Des Moines, IA 50311\\[0.2em]
\textit{E-mail:} \texttt{cp@cpporter.com.}
\textit{Web:} \texttt{\href{http://cpporter.com}{http://cpporter.com}}\par
\vfill \thispagestyle{empty}
\clearpage

\section{Introduction}\label{hhHDcoYGX4}
A well-known way to obtain an algorithmically random stream is via the following thought 
experiment due to Chaitin \cite{MR0411829}. Consider an oracle universal Turing machine that
operates in the binary alphabet and
performs a computation according to the answers it receives to the oracle-queries that it poses.
Such a computation may halt (perhaps with some string written on an output tape) or may continue
indefinitely. If we run the machine, and during the computation we answer the adaptive queries that
appear in a random manner (providing random bits as the answers), then we may consider the probability
that the machine halts.\footnote{Chaitin \cite{MR0411829} originally used self-delimiting machines
in this thought experiment, but our version can easily be seen to be equivalent.}
Chaitin  \cite{MR0411829} showed that the binary expansion of this real number is algorithmically
random, in the classic sense of \ml \cite{MR0223179}.

Subsequent work\footnote{for example, work by
 Becher and her collaborators
\cite{firstBC,fuin/BecherC02,DBLP:journals/jsyml/BecherG05}, 
and the more recent work by the authors in \cite{ranaspro}; see Section \ref{doMNBqznv} for a
discussion.} 
revealed the algorithmic randomness of a variety of probabilities associated
with universal computations, using an approach which establishes the maximum possible
level of algorithmic randomness that is possible given the arithmetical complexity of the associated properties. As we elaborate in Section \ref{doMNBqznv}, this methodology also provided characterizations
of the probabilities in purely algorithmic terms, while a number of examples have been noticed where
it is not applicable. 

{\bf The machine model.} In the present paper we 
consider oracle Turing machines $M$ working on the binary alphabet as input/output devices, which take
as input the oracle stream $X$ and eventually (perhaps after infinitely many steps) 
output a contiguous (\ie without gaps) 
finite or infinite binary sequence $M(X)$
on its one-way output tape.\footnote{Alternatively we could consider as input the pairs $(X,n)$
of oracle binary streams and integers, and the output
as partial function $n\mapsto M(X,n)$ that is computed by $M$ on oracle $X$ and input $n$.
The point here is that the eventual content of the output tape
of $U$ is a representation (\ie the characteristic sequence) of the partial function $n\mapsto U(X,n)$.
The restriction that the output tape is written in a contiguous manner does not make any difference
to our discussion, since we do not consider events that involve partial functions
(something which is no longer true for many of the results in \cite{ranaspro}).}

We are interested in
the probability that the output $U(X)$ of
a universal oracle Turing machine $U$ is a total (\ie infinite) computable
stream, when using a random oracle $X$. Here the oracle is regarded as a random variable
(in the sense of probability, not algorithmic randomness), and the probability of an event
is the measure of oracles $X$ for which this event occurs when $U$ runs with oracle $X$.

\begin{thm}[Main result]\label{DEdI4V18tW}
The probability that the output $U(X)$ of a (universal) oracle 
Turing machine $U$ is total and computable when
reading from a random oracle has the form $\alpha-\beta$, where $\alpha,\beta$ are
\lcep reals. Conversely, given  $\alpha,\beta$ as above such that $\alpha-\beta\in (0,1)$, 
there exists a
universal oracle Turing machine $U$ such that the output $U(X)$ is total and computable
with probability exactly $\alpha-\beta$.
\end{thm}

\begin{rem}[D.c.e.\ reals]
We note that the differences of \lce reals, also known as \dce reals, 
form a field under the usual addition and multiplication, as was
demonstrated by Ambos-Spies, Weihrauch, and Zheng 
\cite{Ambos.ea:00}. Raichev \cite{Raichev:05} and Ng \cite{Ng06} showed that this field is real-closed.
More recently, 
Miller \cite{derivationmiller} developed a theory of derivation on the \dce reals, 
generalizing a result from
\cite{omegax} which will be crucial in the proof of our main result.
The same facts hold by direct relativization for the class of reals
in Theorem \ref{DEdI4V18tW}, namely
 the class of differences of \lcep reals.
\end{rem}

\begin{rem}[Significance of binary output alphabet]
The fact that we restrict our machines to the binary alphabet is crucial for our methods
and even for our results. The motivated reader is referred to 
Becher and Grigorieff \cite[Theorem9.4]{DBLP:journals/jsyml/BecherG09}
which refers to machines with infinite output alphabet and 
contrasts our Theorem \ref{DEdI4V18tW}. For example, the rather simple argument of
Section \ref{wqAZ3YkoA} relies on the fact that our output alphabet is binary.
\end{rem}

It follows from Theorem \ref{DEdI4V18tW} that, although 
the property that $U(X)$ is total and computable  is $\Sigma^0_3$-complete when $U$ is universal, 
its probability can be as simple as 1/2 and as complicated as \ml random or even
\ml random relative to the halting problem.
\begin{coro}
There is a universal oracle Turing machine $U$ such that
the probability that the output $U(X)$ is total and computable is 1/2.
\end{coro}
The contrast here is that the complexity of a universal machine in combination with a 
nontrivial $\Sigma^0_3$-complete property does not necessarily transfer to
the probability of this property. As we are going to explain later on, there are two main
nontrivial reasons for this phenomenon:
\begin{itemize}
\item properties of bases for \ml randomness;
\item properties of differences of universal halting probabilities.\footnote{Here and in the following
we use the term  `universal probability' to refer to the probability that some property holds
when a universal Turing machine runs on a random oracle; see \cite{ranaspro} for various 
examples of universal probabilities.}
\end{itemize}
An oracle $X$ is a base for \ml randomness if it is computed by another oracle
which is \ml random relative to $X$. This class was introduced by Ku\v{c}era \cite{MR1238109}
and has been studied extensively since. First, we will use the fact from 
\HNF \cite{MR2352724} that bases for \ml randomness are computable from the halting problem
in order to show that the probability in Theorem \ref{DEdI4V18tW} is the same as the probability
of another property which has simpler arithmetical complexity, namely that of being defined
as the disjunction of a \pzt and a \szt formula. This will be sufficient for one direction of the
characterization in Theorem \ref{DEdI4V18tW}.
For the other direction we will use a property of differences of universal halting probabilities that was
recently discovered in \cite{omegax}. A simple way to state this property is that if $\alpha,\beta$
are universal halting probabilities and $q>1$ is a rational number, then at least one of 
$\alpha-\beta$, $q\cdot\alpha-\beta$ is \ml random and either a \lce or a \rce real.

\subsection{Overview of this article}
After the background and terminology of Section \ref{V5gJ2M7GQp},
we briefly discuss previous work in the literature regarding
the algorithmic randomness of probabilities of universal machines in
Section \ref{doMNBqznv}.
The point we are making here is that the 
characterization given in our result, Theorem \ref{DEdI4V18tW},
presents a new paradigm in relation to the existing work on the algorithmic 
randomness of probabilities of machines.
In Section \ref{vlfBnKCiKF} we briefly discuss a universal probability in the context of formal systems
and Chaitin's popular metamathematical considerations. We point out that the probability
that a universal machine will produce an undecidable sentence of Peano arithmetic is a 
\ml random \lce real. This probability is of the form $\alpha-\beta$ for two \ml random \lce reals
$\alpha,\beta$, just like the property of Theorem \ref{DEdI4V18tW}. However it behaves
very differently, since it can always be approximated computably from the left
(a consequence of the recent work in \cite{omegax}).
This example provides further context for our main result, and shows that 
a natural property that can be expressed as 
the difference of two $\Sigma^0_n$ classes does not necessarily admit a probability characterization
such as the one given in Theorem \ref{DEdI4V18tW}.
Finally Section  \ref{Md9g6AfPbo} contains the proof of Theorem \ref{DEdI4V18tW}.

\subsection{Background, notation and terminology}\label{V5gJ2M7GQp}
The reader will need to be familiar with the basic notions of computability theory and 
the definition of \ml randomness \cite{MR0223179}.
Let $\emptyset^{(n)}$ be the $n$th iteration of the halting set and let
$\mathbf{0}^{(n)}$ be its Turing degree. Let $\emptyset'$ also denote the halting problem
and let $\mathbf{0}'$ denote the degree of the halting problem.
For each $n>0$ we say that a real is $n$-random if it is \ml random relative to $\emptyset^{(n-1)}$.
A real is called \lce if it is the limit of a computable increasing sequence of rational numbers.
Similarly, a real is called \rce if it is the limit of a computable decreasing sequence of rational numbers.
These notions relativize with respect to $\emptyset^{(n)}$ for each $n>0$. In the present paper
the \lcep and the \rcep reals will be particularly relevant.

Some familiarity with the celebrated 
characterization of halting probabilities as the \ml random \lce reals will also be useful.
This well-known result was a consequence of the cumulative effort of
Solovay \cite{Solovay:75u}, \CHKW \cite{Calude.Hertling.ea:01}, and \KS \cite{Kucera.Slaman:01}.
A summary of basic facts about the representation of open 
$\Sigma^0_n$ classes of streams as $\Sigma^0_n$ \pf sets of
strings and the characterization of their measures as the $\mathbf{0}^{(n)}$-left c.e.\ reals
can be found in \cite{ranaspro}.
If $Q$ is a set of strings, we let $\dbra{Q}$ be the set of streams which have a prefix in $Q$.
Moreover given a class $\mathcal{C}$ of streams we let $\mu(\mathcal{C})$ denote the Lebesgue
measure of $\mathcal{C}$.

As we indicated early in Section \ref{hhHDcoYGX4}, our Turing machine model is the
standard machine $M$ with a one-way read-only input tape,  a working tape and a
one-way write-only output tape which is initially blank and on which the output is printed in binary
contiguously (\ie without leaving blanks between two bits) by a head that moves only to the right.
Here the contents of the input tape $X$ can be treated as an oracle, or alternatively 
as a random variable, in which case we talk about {\em randomized} or
{\em probabilistic} machines and computations. In this case $M(X)$ denotes the 
contents of the output tape when the machine is run indefinitely.
Alternatively, 
the reader may think in terms of
monotone machines such as those used in Levin in \cite{levinthesis,Levin:73} 
in order to give a definition of
the algorithmic complexity of finite objects (a similar notion was used earlier 
by Solomonoff \cite{Solomonoff:64}).
Let $\preceq$ denote the prefix relation amongst strings.
A monotone machine can be thought of as
a Turing machine $M$ 
operating on finite binary programs with the 
monotonicity property that
if $\sigma\preceq \tau$, $M(\sigma)\de$, and $M(\tau)\de$, then $M(\sigma)\preceq M(\tau)$.
In this case, given an infinite binary stream $X$, we let $M(X)$ denote the supremum 
of $M(\sigma)$ for all prefixes $\sigma$ of $X$.
The machines constructed in Section \ref{M2evXEjHpf} are best thought of as monotone machines. 
In any case, the underlying notion is that of an infinitary computation, which is 
performed over an infinite number of stages and where the output is either a binary string or
a binary stream (\ie an infinite binary sequence).
A Turing machine is universal if it can simulate any other Turing machine with a constant
overhead on the input tape. Let $\ast$ denote the concatenation of stings. 
Given an effective list $(M_e)$ of all oracle machines, an 
oracle machine $U$ is universal if 
there exists a computable function $e\mapsto\sigma_e$ from 
numbers into a \pf set of strings such that
$U(\sigma_e\ast X)[s]=M_e(X)[s]$ for all $e, X, s$, where `$[s]$' indicates the
state of the preceding computation after $s$ many steps.
For the case of monotone machines,
the definition of universality is analogous.\footnote{We 
stress that this standard notion of universality is quite different than the notion
of {\em optimal machines} in the context of Kolmogorov complexity (e.g.\ see \cite[Definition 2.1]{MR1438307}).}

Given an oracle Turing machine $M$ we consider properties of oracles $X$ of the type
\[ 
\textrm{$\mathcal{P}(X)$:\ \  the output $M(X)$ belongs to a class $\mathcal{C}$.}
\]
The probability of such a property $\mathcal{P}$ with respect to an oracle machine $M$
is simply the measure of the set of oracles $X$ which have the property $\mathcal{P}$.
When we talk about the {\em universal probability of a property} $\mathcal{P}$
we mean the probability of $\mathcal{P}$ with respect to a universal oracle Turing machine.
This is the formal context for our main result, Theorem \ref{DEdI4V18tW}, which refers to
the property that $M(X)$ generates an infinite computable binary stream as its output.
In Section \ref{doMNBqznv} the reader may find various examples of different
properties $\mathcal{P}$ with respect to which the universal probabilities have been characterized.

More special results and notions will be defined
and cited in the text, when we need them.
For example we make use of relatively recent results from 
\HNF \cite{MR2352724} as well as Barmpalias and Lewis-Pye \cite{omegax}.
For a more comprehensive background in the area between 
computability theory and algorithmic randomness 
we refer the reader to the monographs Downey and
Hirschfeldt 
 \cite{rodenisbook} and Nies \cite{Ottobook}, 
while Calude \cite{caludebook} has an information-theoretic perspective. 
Odifreddi \cite{Odifreddi:89,Odifreddi:99} is a standard reference in classical computability
theory while Li and Vit{\'a}nyi \cite{MR1438307} is a standard reference in the theory of Kolmogorov complexity (which provides another facet of algorithmic randomness).

\section{Algorithmic randomness of probabilities}
We give some context for our main result and demonstrate its uniqueness in this line of research.

\subsection{Previous work---randomness by maximality}\label{doMNBqznv}
Since Chaitin's work many more examples of random numbers have been exhibited as
probabilities of certain properties of various models of universal machines.
A major influence in this line of work was a series of papers by Becher and her collaborators
(\eg \cite{firstBC,fuin/BecherC02,DBLP:journals/jsyml/BecherG05}), while the authors of the present
article recently pushed this line of work to obtain complete characterizations of such probabilities
in terms of algorithmic randomness in \cite{ranaspro}.\footnote{In \cite{ranaspro} the reader will also
find a more detailed summary of the results in this topic.} Although the arguments employed
in these proofs of randomness may seem varied (some expressed in terms of initial segment complexity
and some in terms of statistical tests) they all follow a general paradigm, which we may call
{\em randomness by maximality}. For example, 
\begin{equation*}
\parbox{14.4cm}{given a property $P$ of a certain arithmetical complexity,
one shows that the probability that a universal oracle Turing machine will have $P$ when it runs on
a random oracle is algorithmically random with respect to all statistical tests of the same
arithmetical complexity as $P$.}
\end{equation*}
Note that the arithmetical complexity of the given property automatically imposes an analogous upper 
bound on the level of the algorithmic randomness that the associated probability possesses.
Hence in this methodology one shows that the probability of the property $P$ is
{\em as algorithmically random as it can possibly be, given its arithmetical complexity}, hence the name
{\em randomness by maximality}.

Such a plan can often be carried out, as it is elaborated in  \cite{ranaspro}, 
by embedding a member of the universal \ml test of the respective arithmetical complexity
into an oracle Turing machine $M$ in such a way that the class of oracles $X$ that make 
the computation $M(X)$ satisfy property $P$ is identical to the class of oracles $X$
that fail the respective \ml test.\footnote{For example, in order to show that the halting probability
(as described in the example above) is \ml random, one can consider a member $V$ 
of the universal \ml test and a machine $M$
which halts exactly on the oracles in $V$. Then the halting probability $\alpha$ of $M$ will be the
measure of $V$, which is \ml random. Since $M$
will be simulated by the universal oracle Turing machine $U$, the halting probability
of $U$ will be the sum of the \ml random number $\alpha$ and another real of similar complexity,
which (by known facts) has to be \ml random.}
This methodology (and its counterpart in terms of initial segment complexity) has produced
many results of the form
\begin{equation*}
\parbox{14.4cm}{the probability of property $P$ which lies at the $n$th level of arithmetical complexity is 
algorithmically random with respect to all statistical tests of the same arithmetical complexity.}
\end{equation*}
For example, halting is a $\Sigma^0_1$ property, and the halting probability of a universal machine
is random with respect to all unrelativized (hence $\Sigma^0_1$) \ml tests, also called 1-random.
Similarly, $\Sigma^0_2$ properties (\eg infinite/finite output) 
often produce 2-random probabilities (\eg \cite{firstBC,ranaspro}),
$\Sigma^0_3$ properties (\eg almost everywhere totality of computed function) 
often produce 3-random probabilities (\eg \cite{fuin/BecherC02,ranaspro})
and some $\Sigma^0_4$ properties (\eg universality) 
produce 4-random probabilities \cite{Barmpalias3488}.
In most of these cases, this maximality approach also allows
for a complete characterization
of the probabilities considered in terms of 
algorithmic definability and randomness.

Despite all these examples, it has been noticed on several occasions that the 
probability of a property with respect
to the universal Turing machine may not necessarily have the
maximum
algorithmic randomness that  its arithmetical complexity allows.
For example, \BFGM \cite{jsyml/BecherFGM06}
showed that for any $n>1$, 
the probability that the universal oracle Turing machine will halt
and output a string in a given $\Sigma^0_n$-complete set 
(this is a $\Sigma^0_n$-complete property)
is not $n$-random. 
A few more such counterexamples are briefly discussed in \cite{ranaspro}.

\subsection{Restricted halting probabilities and differences}\label{vlfBnKCiKF}
In this section we discuss some metamathematical issues that are related to
universal halting probabilities. During this discussion the machine model is
the oracle Turing machine, where the oracle $X$ is regarded as the input
and the output is a finite string, which is produced after a computation has halted.\footnote{This contrasts
the main part of the paper where we consider infinitary computations.}
Following Chaitin's widely publicized metamathematical
considerations (mainly based on \cite{MR0411829,Chaitin:1992:II}), 
part of the appeal of the halting probability comes from
its connections with formal systems of mathematics. For example, if
we effectively identify binary strings with statements in Peano arithmetic,
the set $A$ of theorems of Peano arithmetic is a \ce set of strings. Chaitin 
\cite{chaitin2004algorithmic} observed
that in this case the probability $\Omega(A)$ that the universal machine will
produce a theorem of Peano arithmetic is a \lce \ml random number.
Let us now consider the set $B$ of arithmetical sentences which are undecidable in
the formal system of Peano arithmetic, \ie neither they nor their negation is provable 
from the axioms of Peano arithmetic. The set $B$ is \pz and by G\"{o}del's incompleteness theorem, it
is nonempty. Which properties are satisfied by the probability $\Omega(B)$ that the randomized 
universal machine will produce an undecidable sentence? In \cite{omegax}, 
answering a problem from  \cite[Question 8.10]{MR2248590}
and \cite{DBLP:journals/jsyml/BecherG05,jsyml/BecherFGM06}, 
it was shown that for any nonempty \pz set $B$, the number $\Omega(B)$
is both \lce and \ml random. Applying this result to the set of undecidable sentences, we get
 the following counterintuitive fact.
\begin{thm}\label{6FcQ8QWaaZ}
The probability that the randomized universal machine outputs an undecidable 
sentence is a \lce \ml random real.
\end{thm}
The contrast here is that, although we cannot effectively enumerate the undecidable sentences,
we can enumerate the left Dedekind cut of the probability of an undecidable output, just like
we did for the probability of a theorem of Peano arithmetic  (the latter being computably enumerable).
If we recall the characterization of universal halting probabilities as the \ml random \lce reals, then
it follows that the probability of an undecidable sentence with respect to any universal machine $U$
is equal to the probability of a theorem of Peano arithmetic with respect to another universal machine $V$.

We can continue this discussion briefly by considering the probability that the randomized universal
machine will output a  true arithmetical sentence. 
\BFGM \cite{jsyml/BecherFGM06} showed that the
halting probability restricted to any $\Sigma^0_n$-complete set is \ml random\footnote{It is worth 
noting that Kobayashi \cite{ipl/Kobayashi93} had obtained a weaker version of this result for
a restricted notion of $\Sigma^0_n$-completeness which he called {\em constant overhead completeness}.
We also note that the result in \BFGM \cite{jsyml/BecherFGM06} applies for optimal machines in general,
not just universal machines.} and not $n$-random. Since the set of 
$n$-quantifier arithmetical sentences is an example of a $\Sigma^0_n$-complete set, we have the 
following.
\begin{thm}[\BFGM \cite{jsyml/BecherFGM06}]\label{DAJRioO5S}
For each $n>1$, the probability that the randomized universal machine outputs a true 
$n$-quantifier arithmetical sentence is \ml random but not $n$-random.
\end{thm}
Two questions present themselves at this point. First,
what about the probability that  the randomized universal machine outputs a true
arithmetical sentence (of any number of quantifiers)? Is it algorithmically random, and if yes,
how much? Second, is the probability in Theorem \ref{DAJRioO5S} \ml random with respect to
$\emptyset'$ or $\emptyset''$
for sufficiently large $n$? These questions do not seem to be amenable to the analysis in
\cite{jsyml/BecherFGM06}.
We leave them open for consideration, and  move on to
our main topic.

This discussion is interesting as a continuation of Chaitin's metamathematical considerations,
but its main purpose is to provide some additional context for our main result, Theorem \ref{DEdI4V18tW}.
In particular, Theorem \ref{6FcQ8QWaaZ} shows that even in situations where
the maximality paradigm of Section \ref{doMNBqznv} does not apply (\eg when the property in question
is definable by the disjunction of a $\Sigma^0_n$ and a $\Pi^0_n$ formula) it could still be the 
case that its universal probability is {\em necessarily} maximally random (in the case of a
$\Sigma^0_n\vee\Pi^0_n$ formula this means $n$-random).\footnote{Theorem \ref{6FcQ8QWaaZ}
is an example of this phenomenon for $n=1$, while a relativization 
of it produces examples for any $n$.}
Hence Theorem \ref{DEdI4V18tW} does not have a precedent in this line of research,
and all these examples reveal the wide variety of algorithmic behavior that the probability
of a property of the universal machine can have.

\section{Proof of Theorem \ref{DEdI4V18tW}}\label{Md9g6AfPbo}
First, in Section \ref{wqAZ3YkoA}, we show the first part of  Theorem \ref{DEdI4V18tW}, namely 
that we can express the universal probability of a computable output as the difference of two
\lcep reals. Section \ref{M2evXEjHpf} contains the first step towards the converse of our main result.
We show that  the probability of a computable output  can be any \lcep or \rcep real provided
that we choose an appropriate machine. Note however that we do not yet guarantee this 
for a universal Turing machine. The first step toward the latter conclusion is made in
Lemma \ref{INsGDxV55i} of the next section, which shows that the universal 
probability of a computable output can be chosen to be any \ml random \lcep real.
We conclude Section \ref{fuoAC5uTeq} with the proof of the latter part of Theorem 
 \ref{DEdI4V18tW}, in an argument that uses all the previous lemmas.
 
\subsection{The probability of a computable output as the difference of \texorpdfstring{\lcep}{0'-l.c.e.}reals}\label{wqAZ3YkoA}
We are interested in the measure of the oracles $X$ such that
the output $M(X)$ is total and computable.
Let $\tot(M)$ denote the $\Pi^0_2$ class of streams $X$ such that $M(X)$ is total.
Also let $\inctot(M)$ be the class of streams $X$ such that $M(X)$ is total and incomputable.
Then we are interested in the measure of the class
\begin{equation}\label{DBiC4T8TdE}
\tot(M)-\inctot(M).
\end{equation}
We note that $\inctot(M)$ is not generally definable with 2 quantifiers, so that the above class
appears more complex than we would ideally wish. Despite this, we show that the measure of this 
class can be expressed as the difference of two \lcep reals.

Recall that an oracle $X$ is a base for \ml randomness if it is computed by another oracle
which is \ml random relative to $X$. Also note that by
\HNF \cite{MR2352724} bases for \ml randomness are computable from the halting problem.
On the other hand, 2-random streams do not compute any incomputable set which is also computable
from the halting problem. This latter statement follows from the fact that the class of oracles that compute
a $\Delta^0_2$ set $X$ form a $\Sigma^0_3$ class, which is null when $X$ is noncomputable 
(by \cite{deleeuw1955}, also see \cite{Downey.Nies.ea:06}). These facts will be used in the proof
of the following lemma.
 
\begin{lem}\label{wi7Kr71TMX}
The probability that an oracle Turing machine $M$ produces a computable output when
reading from a random oracle has the form $\alpha-\beta$, where $\alpha,\beta$ are
\lcep reals in $(0,1)$.
\end{lem}
\begin{proof}
If $\alpha,\beta$ are \lcep reals in $(0,1)$, let $q\in (\max\{\alpha,\beta\}, 1)$ be a rational number
and note that
\[
\alpha-\beta= 
-\beta - (-\alpha)=
(q-\beta) - (q-\alpha).
\]
Hence the difference of two \lcep reals in $(0,1)$ can be written as the difference of two
\rcep reals in $(0,1)$.
Hence it suffices to show that the measure of \eqref{DBiC4T8TdE} 
can be written as the difference of two
\rcep reals in $(0,1)$.

Note that while $\tot(M)$ is a $\Pi^0_2$ class,
$\inctot(M)$  is a $\Pi^0_3$ class, so we need to find a simpler class which has the same
measure as $\inctot(M)$. Let $(V_i)$ be a universal \ml test and define
\[
\inctot^{\ast}(M)=\tot(M)\cap\{X\ |\ X\in \cap_i V^{M(X)}_i\}.
\]
This is the class of streams $X$ such that $M(X)$ is total and $X$ is not 
\ml random relative to $M(X)$.
Observe that $\inctot^{\ast}(M)$ is a $\Pi^0_2$ class.
If we show that
\begin{equation}\label{xFQqbEWMud}
\textrm{for every 2-random $X$ we have 
$X\in\inctot(M)\Leftrightarrow X\in\inctot^{\ast}(M)$} 
\end{equation}
then $\mu(\inctot(M))=\mu(\inctot^{\ast}(M))$
which means that
\[
\mu\big(\tot(M)\big)-\mu\big(\inctot^{\ast}(M)\big)
\]
is the measure of \eqref{DBiC4T8TdE}.
Since the above two classes are $\Pi^0_2$, their measures are \rcep reals, so this
proves the statement. It remains to prove \eqref{xFQqbEWMud}.

First, suppose that $X$ is 2-random and $X\in\inctot$. 
Then by definition $M(X)$ is total and incomputable. If
$X\not\in \cap_i V^{M(X)}_i$ then $X$ would be \ml random relative to the incomputable set $M(X)$
which it computes, so $M(X)$ would be a base for \ml randomness.  But we know from 
\cite{MR2352724} that such sets
are $\Delta^0_2$ and we also know that no 2-random real computes a noncomputable
$\Delta^0_2$ set. Hence we arrived at a contradiction, which means that 
$X\in \cap_i V^{M(X)}_i$.

Conversely, suppose that $X$ is 2-random, $M(X)$ is total and $X\in \cap_i V^{M(X)}_i$. Then
$M(X)$ must be incomputable, otherwise $X$ would not be 2-random (or even 1-random).
So $X$ belongs to $\inctot$.
This concludes the proof of \eqref{xFQqbEWMud} and the proof of the lemma.
\end{proof}

\subsection{A Turing machine for each \texorpdfstring{\lcep}{left 0'-c.e.}and 
each \texorpdfstring{\rcep}{right 0'-c.e.}real}\label{M2evXEjHpf}
Here we make the first step towards the proof of the second part of Theorem
\ref{DEdI4V18tW}. 
\begin{defi}[Weight of a set of strings]
The weight of a \pf set of strings $S$ is $\sum_{\sigma\in S} 2^{-|\sigma|}$.
\end{defi}

We wish to show that given any $\alpha$ which is
either a  \lcep or a \rcep real, there exists an oracle Turing machine $M$ such that
the measure of $\tot(M)-\inctot(M)$ is $\alpha$.
In order to do this we use the fact that 
every  \lcep real is the weight of
a $\Sigma^0_2$ \pf set of strings.\footnote{A proof of this fact can be found in 
\cite[Section 2.1]{ranaspro}.}
Given such a set of strings, we wish to produce a special oracle Turing machine $M$
that has the desired property (namely $M(X)$ being total and computable)
exactly on streams that have (alternatively those which do not have) a prefix in the
$\Sigma^0_2$ \pf set of strings. Although its not possible to achieve this in both cases, 
we will be successful almost everywhere, in the probabilistic sense, which is sufficient for our purposes.

We need the following fact.
\begin{lem}[Canonical $\Sigma^0_2$ approximations]\label{SBZGB7jnF2}
If $U$ is a $\Sigma^0_2$ \pf set of strings,
there exists
a computable sequence $(V_s)$ of finite \pf sets
of strings  such that
\begin{enumerate}[(i)]
\item for each $\sigma$ we have $\sigma\in U$ if and only if there exists $s_0$ 
such that $\sigma$ has a prefix in all $V_s$, $s>s_0$;
\item there are infinitely many $s$ such that $V_s\subseteq U$.
\end{enumerate}
We call $(V_s)$  a {\em canonical $\Sigma^0_2$ approximation to $U$}. 
\end{lem}

\begin{proof}
Given $U$,
there exists a \ce operator $W$ such that $W^{\emptyset'}=U$. 
We can modify the enumeration of $W$ (obtaining a modified $\widehat{W}$) 
with respect to a computable enumeration
$(\emptyset'_s)$ of $\emptyset'$, so that if $n\in \emptyset'_{s+1}-\emptyset'_s$ for some number $n$ 
and stage $s$, any number $m$ which is in $W^{\emptyset'_{s+1}}_{s+1}$ with oracle use above $n$
is not counted in $\widehat{W}^{\emptyset'_{s+1}}_{s+1}$. In other words, the enumeration of 
$\widehat{W}^{\emptyset'}$ follows the enumeration of  $W^{\emptyset'}$ except that it delays the enumeration
of certain numbers until a stage where the associated segment of $\emptyset'$ remains stable between the
current and the previous stages.\footnote{This is a standard technique which is known as the
{\em hat-trick} (originally due to Lachlan) and is applied to functionals and \ce operators relative to \ce sets 
(see Odifreddi \cite[Section X.3]{Odifreddi:99} for an extended discussion on this method).}
Let $U_s=\widehat{W}^{\emptyset'_{s}}_{s}$
and let $V_s$ contain the minimal strings in $U_s$. 
Clearly all $V_s$ are finite and \pf sets
of strings. 
A {\em true} enumeration into $\emptyset'$ is 
an enumeration of a number $n$ into $\emptyset'$ at stage $s$
such that $\emptyset'_{s}\restr_n=\emptyset'\restr_n$. A {\em true stage} is a stage $s$ at which 
a true enumeration occurs. Clearly there exist infinitely many true enumerations and stages. 
By the choice of $\widehat{W}$ we have that if $s$ is a true stage then  
$U_s\subseteq U$ and
since $U$ is \pf so is $U_s$ and thus we have $V_s=U_s\subseteq U$.
Moreover, if $\sigma\in U$ then $\sigma\in U_s$ for 
all but finitely many stages $s$, which means that $\sigma$ has a prefix in $V_s$ for 
all but finitely many $s$. Finally, by the hat trick, if
$\sigma$ has a prefix in $V_s$ for all but finitely many stages $s$, 
we necessarily have $\sigma\in U$,
because $V_t\subseteq U$ for infinitely many stages $t$.
\end{proof}

\begin{lem}\label{T9RUBd4GeX}
Suppose that $Q$ is a \szt  \pf set of strings. 
Then there exists an oracle Turing machine $M$
such that $M(X)$ is total for all $X$, and,  
for every \ml random real $X$, the following clauses are equivalent: 
\begin{enumerate}[\hspace{0.5cm}(a)]
\item $M(X)$ is total and computable;
\item $M(X)$ is total  and $M(X)=\eta\ast 0^{\omega}$ for some string $\eta$;
\item $X\in\dbra{Q}$;
\end{enumerate}
\end{lem}
\begin{proof}
Given a \szt  \pf set of strings
$Q$ we use Lemma
\ref{SBZGB7jnF2} and consider
a canonical \szt approximation $(Q_s)$ to $Q$.
Note that a string $\sigma$ belongs to $Q$ if and only if
it has a prefix in  $Q_s$ for almost all $s$. 
We will use $(Q_s)$ in order to build an oracle machine
$M$ with the required properties. We may assume that 
for each $s$, the strings in $Q_s$
have length less than $s$. 
We build $M$ as a monotone machine.
Let $\lambda$ denote the empty string. 
At stage 0, we define $M(\lambda) = \lambda$.  At stage $s+1$ we define
$M(\sigma)$ for all strings of length $\sum_{i\leq s} 2^i=2^{s+1}-1$.
Hence at stage 1, we define $M(\sigma)$ for all strings $\sigma$ of length 1, at stage 2, 
we define $M(\sigma)$ for all strings of length 3, and so on.
Moreover we ensure that the length  of $M(\sigma)$
is equal to the length of $\sigma$. Hence at stage $s+1$, for each string $\sigma$
of length $2^{s+1}-1$ with a prefix $\tau$  of length $2^{s}-1$ we need to determine 
$2^s$ additional bits which we can suffix to $M(\tau)$ in order to define $M(\sigma)$.

At stage $s+1$ do the following for each string $\sigma$ of length 
$2^{s+1}-1$ and its prefix $\tau$ of length $2^{s}-1$.
If there is a prefix of $\sigma$ in $Q_{s}$, then let $M(\sigma)=M(\tau)\ast 0^{2^s}$.
Otherwise let $M(\sigma)=M(\tau)\ast \rho$ where $\rho$ is the string such that
$\sigma\restr_{|M(\tau)|}\ast\rho=\sigma$, \ie the last $|\sigma|-|M(\tau)|$ many bits of
$\sigma$. This completes the construction of $M$.

Clearly $M$ is a monotone machine and $M(X)$ is total for all streams $X$.
Moreover by the construction and the properties of $(Q_s)$ we have that
\begin{equation*}
\parbox{13cm}{if $\sigma\in Q$ then for each $X$ extending $\sigma$
there exists a string $\rho$ such that $M(X)=\rho\ast 0^{\omega}$.}
\end{equation*}
Indeed, if $\sigma\in Q$ then there exists a stage $s_0>0$ such that $\sigma\in Q_s$ for all
$s\geq s_0$. In this case for each $X$ that has prefix $\sigma$ 
the construction gives that $M(X)=\rho\ast 0^{\omega}$ for some string 
$\rho$ of length $2^{s_0-1}-1$.

Recall that there are infinitely many $s$ such that $Q_s\subseteq Q$. Hence
if $X$ does not have a prefix in $Q$, there are infinitely many stages $s$ such that
$X$ does not have a prefix in $Q_s$.
Let $(\varphi_e)$ be a computable enumeration of all partial computable functions 
and for each $e,s$ let
$V_{s+1}(e)$ be
\begin{itemize}
\item the set of strings $\eta$ of length $2^{s+1}-1$ such that $\eta(i)=\varphi_e(i)$ for all
$i\in [2^{s}, 2^{s+1})$, if $\varphi_e\restr_{2^{s+1}}$ is defined;
\item the empty set, otherwise.
\end{itemize}
By the construction we have $\mu (V_{s+1}(e))\leq 2^{-s}$ for each $s,e$,
and $V_{s+1}(e)$ is uniformly \ce in $e,s$.
Hence $(V_{s+1}(e))$ is a \ml test for each $e$. Hence if $X$ is \ml random and $e\in\Nat$,
then $X$ has a prefix in $V_{s+1}(e)$ 
for only finitely many $s$.

Now we can show that if $X$ is \ml random and does not have a prefix in $Q$ then
$M(X)$ is not computable. Indeed, if $\varphi_e=M(X)$ then at each stage $s+1$
such that $X$ does not have a prefix in $Q_s$, the construction would define
the last $2^s$ many digits of $M(X\restr{2^{s+1}-1})$ to be
the last $2^s$ many digits of $X\restr{2^{s+1}-1}$, which means that the latter
coincides with the last $2^s$ many digits of $\varphi_e\restr_{2^{s+1}-1}$.
But this means that $X\in V_{s+1}(e)$.
Hence by the assumption that $\varphi_e=M(X)$ and $X\notin\dbra{Q}$
we deduce that $X\in\dbra{V_{s+1}(e)}$ for infinitely many $s$. Since 
$(V_{s+1}(e))$ is a \ml test  this means that $X$ is not \ml random.
We have shown that 
\begin{equation*}
\parbox{13cm}{if $X$ is \ml random and $X\not\in \dbra{Q}$ then $M(X)$ is
not computable.}
\end{equation*}
Hence for almost all $X$ in the complement of $\dbra{Q}$ we have that $M(X)$ is
noncomputable, while for all $X$ in $\dbra{Q}$ the image $M(X)$ is computable.
This concludes the proof of the lemma.
\end{proof}

By the Kraft-Chaitin theorem it follows 
(see \cite{ranaspro} for a detailed proof) that
\begin{equation}\label{PHPIxTuwCA}
\parbox{13cm}{given any $e\in\Nat$, any string $\rho$ of length $e$ and any
\lcep real $\alpha\in (0,1-2^{-e})$ there exists  a
\szt \pf set $S$ of strings which are incompatible with $\rho$
and $\mu(\dbra{S})=\alpha$.}
\end{equation}
Here incompatibility is with regard to the prefix relation: a string $\eta$ is incompatible
with a string $\rho$ if  $\eta\neq \rho$ and
$\eta$ is neither a prefix nor an extension of $\rho$ .
Lemma \ref{T9RUBd4GeX} in combination with \eqref{PHPIxTuwCA} 
implies the following item, which will be used in the
argument of Section \ref{fuoAC5uTeq}.

\begin{coro}\label{nyRrTKweTO}
If $e\in\Nat$, $\rho$ is a string of length $e$ and $\alpha$ is a \lcep real in $(0,1-2^{-e})$, 
then there exists an oracle Turing machine $M$  
such that $M(\sigma)$ is the empty string for any
string $\sigma$ which is compatible with $\rho$, and such that the probability
that $M(X)$ is computable is exactly $\alpha$.
\end{coro}

Finally we want a corresponding statement for \rcep reals, which we can get from
an analogue to Lemma \ref{T9RUBd4GeX} for $\Pi^0_2$ sets. 
Note that this case is much simpler and gives a stronger result
(with the claimed equivalence being satisfied by {\em every} real), 
which we formulate as follows.
\begin{lem}\label{RyJmrfp9i1}
Suppose that $Q$ is a  \szt \pf set of strings. 
Then there exists an oracle Turing machine $M$
such that the following clauses are equivalent
\begin{enumerate}[\hspace{0.5cm}(a)]
\item $M(X)$ is total;
\item $M(X)$ is total and $M(X)=0^{\omega}$;
\item $X\not\in\dbra{Q}$;
\end{enumerate}
for each real $X$.
\end{lem}
\begin{proof}
Let $(Q_s)$ be a canonical enumeration of $Q$ and construct $M$ as follows.

At stage $s+1$, for each $\sigma$ which does not have a prefix in $Q_s$
define $M(\sigma)=0^{|\sigma|}$.

Clearly $M$ is monotone, hence well-defined.
If $X$ has a prefix $\sigma$ in $Q$ then there exists a stage $s_0$ such that
$\sigma\in Q_s$ for all $s\geq s_0$. In this case $M(X\restr_n)$ will not be defined
for any $n\geq s_0$, so $M(X)$ is not total. For the other direction, assume that 
$X$ does not have a prefix in $Q$. Then there are infinitely many stages $s$
such that $X$ does not have a prefix in $Q_s$; for such $s$, 
the construction will define $M(X\restr_s)=0^s$ at stage $s+1$. Hence $M(X)$ is total
and equal to $0^{\omega}$. Finally if $M(X)$ is total, then necessarily
$M(X)=0^{\omega}$, which concludes the proof of the equivalence.
\end{proof}

The analogue of 
\eqref{PHPIxTuwCA} for \rcep reals is as follows.
\begin{equation}\label{no5ff5lqaw}
\parbox{13.5cm}{given any $e\in\Nat$, any string $\rho$ of length $e$ and any
\rcep real $\beta\in (2^{-e}, 1)$
there exists  a
\szt \pf set $S$ of strings which are incompatible with $\rho$
and $1-\mu(\dbra{S})=\beta$.}
\end{equation}
Note that since $1-\beta$ is a \lcep real whenever
$\beta$ is a \rcep real,  \eqref{no5ff5lqaw} is a direct consequence of \eqref{PHPIxTuwCA}.
Lemma \ref{T9RUBd4GeX} in combination with \eqref{no5ff5lqaw}
implies the following item, which will be used in the
argument of Section \ref{fuoAC5uTeq}. 

\begin{coro}\label{L4IxWvmBj2}
If $e\in\Nat$, $\rho$ is a string of length $e$ and $\beta$ is a \rcep real in $(2^{-e},1)$, 
then there exists an oracle Turing machine $M$ 
such that $M(\sigma)$ is the empty string, for any
string $\sigma$ which is compatible with $\rho$, and such that the probability
that the output $M(X)$ of $M$ on oracle $X$ is total and computable is exactly $\beta$.
\end{coro}
We keep Corollaries
\ref{nyRrTKweTO} and \ref{L4IxWvmBj2}
and use them as ingredients in the argument of the next section.

\subsection{A universal machine for each difference of \texorpdfstring{\lcep}{left c.e.}reals}\label{fuoAC5uTeq}
The next step towards the proof of the second part of Theorem \ref{DEdI4V18tW}
is to produce {\em universal} oracle Turing machines with prescribed
probability of a computable outcome. This is the main difference with the previous sections,
and this is where the use of \ml randomness is most essential.
We are two steps away from the final proof. First, we deal with the specific task of
producing a universal machine such that the relevant probability is any given 2-random 
\lcep real in the unit interval (by symmetry the same holds for any given 2-random \rcep real).

In the proof below we use the following fact by \KS \cite{Kucera.Slaman:01}:
\begin{equation}\label{YODrsFjkfP}
\parbox{13cm}{given any \lcep 2-random real $\alpha$
and any \lcep real $\beta$, there exists $e_0\in\Nat$ such that
$\alpha - 2^{-e}\cdot\beta$ is a \lcep real for all $e\geq e_0$.}
\end{equation}
This fact was proved in \cite{Kucera.Slaman:01} for 
\lce reals and 1-randomness, but it readily relativizes to any oracle. 
The argument below is based on an idea from 
\CHKW \cite{Calude.Hertling.ea:01}.

\begin{lem}\label{INsGDxV55i}
Given any 2-random \lcep real $\gamma\in (0,1)$, 
there exists a universal oracle machine $V$ such that the probability
that the output $V(X)$ of $V$ with oracle $X$ is total and computable is $\gamma$.
\end{lem}
\begin{proof}
Let $V_0$ be a universal oracle machine. Then by 
 Lemma
\ref{wi7Kr71TMX} there exist \lcep reals $\alpha,\beta\in (0,1)$
such that $V_0(X)$ is computable with probability $\alpha-\beta$.
By \eqref{YODrsFjkfP} there exists some $e$ such that 
\begin{itemize}
\item $\gamma-2^{-e}\cdot \alpha$ is a \lcep real;
\item $\gamma-2^{-e}\cdot (\alpha-\beta)\in \Big(0, 1-2^{-e}\Big)$.
\end{itemize}
In this case note that $\gamma-2^{-e}\cdot (\alpha-\beta)$ is a \lcep real in $(0,1)$.
By Corollary \ref{nyRrTKweTO} consider an oracle Turing machine $M$ such that
$M(\sigma)$ is undefined for all $\sigma$ which are compatible with $0^e$
and the probability that $M(X)$ is total and computable is 
$\gamma-2^{-e}\cdot (\alpha-\beta)$. Then for each $\sigma$ define
$V(0^e\ast\sigma)=V_0(\sigma)$ and for every $\tau$ which is incompatible with $0^e$ 
define $V(\tau)=M(\tau)$. Then the probability that  $V(X)$ is total and computable 
is the sum of probability that $M(X)$ is total and computable, plus 
$2^{-e}$ times the probability that $V_0(X)$ is total and computable. Hence
$V(X)$ is total and computable with probability
\[
\gamma-2^{-e}\cdot (\alpha-\beta) + 2^{-e}\cdot (\alpha-\beta)=\gamma
\]
which completes the proof of the lemma.
\end{proof}

In the following final step of our argument we are going to use 
a fact from Rettinger and Zheng \cite{Rettinger2005} which
also holds in relativized form as follows:
\begin{equation}\label{LEvfauct3I}
\parbox{10cm}{If $\alpha,\beta$ are \lcep reals and 
$\alpha-\beta$ is 2-random, then $\alpha-\beta$ is either
a \lcep real or a \rcep real.}
\end{equation}
Another observation we need is that 
the sum of a 2-random \lcep real with any \lcep real is 2-random.
This is a relativization of the fact, originally proved by Demuth \cite{Dempseu},
that the sum of a 1-random \lce real with any \lce real is 1-random.
We also note that
\begin{equation}\label{nKZW9e14Lf}
\parbox{11cm}{any difference of two
\lcep reals can be written as a difference of
two 2-random \lcep reals.}
\end{equation}
In other words,
for every \lcep reals $\alpha,\beta$ there exist
2-random \lcep reals $\alpha_0,\beta_0$ such that 
$\alpha-\beta=\alpha_0-\beta_0$. In order to see this, consider 
a 2-random \lcep real $\gamma$,
and take
$\alpha_0=\alpha+\gamma$ and 
$\beta_0=\beta+\gamma$. Then \eqref{nKZW9e14Lf} follows
since the sum of a 2-random \lcep real with any \lcep real is 2-random.\footnote{Here we note that,
quite surprisingly, the converse of \eqref{nKZW9e14Lf} does not hold. In particular, 
Miller \cite{derivationmiller} showed that 
there are $\mathbf{0}'$-\dce  reals $\gamma$ 
which are irreducible in the sense that
for all \lcep reals $\alpha,\beta$ such that $\gamma=\alpha -\beta$ we have that
$\alpha,\beta$ are 2-random.}

The key ingredient of the proof of Lemma \ref{OubpJyiirE} is a
result from \cite{omegax}, which says that the difference between a \dce real which is not 
1-random and a 1-random \lce real is a 1-random \rce real.  In relativized form this says that
\begin{equation}\label{dVAtWybk5k}
\parbox{11cm}{If $\alpha,\beta,\gamma$ are \lcep reals, $\gamma$ is 2-random 
and $\alpha-\beta$ is not 2-random, then $(\alpha-\beta)-\gamma$ is a 2-random \rcep real.}
\end{equation}
Finally we also use the elementary fact that if $\alpha$ is a 2-random \lcep real then for every
$e$ the number $2^{-e}\cdot \alpha$ is a  2-random \lcep real.

\begin{lem}\label{OubpJyiirE}
Let $\alpha,\beta$ be \lcep reals  such that  $\alpha-\beta\in (0,1)$.
Then there exists a universal Turing machine $M$ such that $M(X)$ is computable
with probability exactly $\alpha-\beta$.
\end{lem}
\begin{proof}
By \eqref{nKZW9e14Lf}, 
without loss of generality we may assume that $\alpha,\beta$ are 2-random.
First we consider the case where
$\alpha-\beta$ is 2-random. In this case by \eqref{LEvfauct3I} 
$\alpha-\beta$ is either a \lcep or a \rcep real. In the first case
the result follows from
Lemma \ref{INsGDxV55i}. Otherwise 
$\alpha-\beta$ is  \rcep and we can use Lemma \ref{INsGDxV55i}
to get a universal Turing machine
$V$ such that the probability of $V(X)$ being computable 
is a \lcep real $\delta$. Also let $e$ be a number such that
\[
(\alpha-\beta) - 2^{-e}\cdot\delta \in \Big( 2^{-e}, 1\Big)
\]
and note that the real $(\alpha-\beta) - 2^{-e}\cdot\delta$ is a \rcep real 
with respect to which
the number $e$ satisfies the condition of Corollary \ref{L4IxWvmBj2}.
So we can use Corollary \ref{L4IxWvmBj2}
in order to get a machine $N$
such that $N(\sigma)$ is undefined for any string $\sigma$
which is compatible with $0^{e}$ (namely the string consisting of $e$ many 0s) 
and such that
the probability that $N(X)$ is computable is
$(\alpha-\beta) - 2^{-e}\cdot\delta$.
Now we can define $M(\rho)=N(\rho)$ for every
$\rho$ which is incompatible with $0^e$, and
$M(0^{e}\ast\sigma)=V(\sigma)$ for all $\sigma$.
Then the probability that $M(X)$ is computable is the sum
of the probability that $N(X)$ is computable, plus $2^{-e}$ times
the probability that $V(X)$ is computable. Hence $M(X)$ is computable with
probability
\[
\big((\alpha-\beta) - 2^{-e}\cdot\delta\big) + 2^{-e}\cdot \delta = \alpha-\beta.
\]
It remains to consider the case where
$\alpha-\beta$ is not 2-random. 
Consider the universal machine $V$ as above, so that the
probability of $V(X)$ being computable is a 2-random \lcep real $\delta$.
In this case by \eqref{dVAtWybk5k} the number $\alpha-\beta -2^{-e}\cdot\delta$ is \rcep
for every $e$. Let $e$ be such that $\alpha-\beta -2^{-e}\cdot\delta\in (2^{-e},1)$
and use Corollary \ref{L4IxWvmBj2} in order to obtain a machine $N$ such that
$N(\sigma)$ is undefined for any string $\sigma$ which is compatible with $0^{e}$
and $N(X)$ is computable with probability 
$\alpha-\beta -2^{-e}\cdot\delta$. Now define $M(\sigma)=N(\sigma)$
for every  string $\sigma$ which is incompatible with $0^{e}$,
and $M(0^e\ast \rho)=V(\rho)$ for all $\rho$.

Then the probability that $M(X)$ is computable is the sum
of the probability that $N(X)$ is computable, plus $2^{-e}$ times
the probability that $V(X)$ is computable. Hence $M(X)$ is computable with
probability
\[
(\alpha-\beta -2^{-e}\cdot\delta) + 2^{-e}\cdot\delta = \alpha-\beta
\]
which concludes the proof of the lemma.
\end{proof}
This completes the proof of Theorem 1.1. 

\section{An application of our analysis}
The methodology we developed in Section \ref{Md9g6AfPbo} 
is applicable to other problems in this area. As a demonstration,
we give the following characterization of the universal probabilities that the output ends in a stream of
0s, as a corollary of the previous analysis. 
\begin{thm}\label{OduFLoeqO}
If $V$ is a (universal) oracle  Turing machine then the probability that $V(X)=\rho\ast 0^{\omega}$ 
for some string $\rho$ is $\alpha-\beta$ for some 
\lcep reals $\alpha,\beta$. Conversely, given  $\alpha,\beta$ as above such that 
$\alpha-\beta\in (0,1)$, 
there exists a
universal oracle Turing machine $V$ such that the probability that  $V(X)=\rho\ast 0^{\omega}$
for some string $\rho$ is $\alpha-\beta$.
\end{thm}
The first part follows from the part that the class in question is the difference of two \pzt classes. 
For the second part note that
Lemma \ref{T9RUBd4GeX} and Lemma \ref{RyJmrfp9i1} imply that the probability that $M(X)$ is computable
equals the probability that $M(X)=\rho\ast 0^{\omega}$ for some string $\rho$.
Therefore Corollary \ref{nyRrTKweTO} and Corollary \ref{L4IxWvmBj2}
hold for the probability that $M(X)=\rho\ast 0^{\omega}$ for some string $\rho$,
in place of the probability that $M(X)$ is computable. This means that the machines $M$
constructed in the various cases of the proof of Lemma \ref{OubpJyiirE}
also prove the second clause of Theorem \ref{OduFLoeqO}.

\section{Concluding remarks}
We have characterized the universal probabilities of a computable output as the class
of differences of \lcep reals in $(0,1)$. Our methodology, as well as the characterization itself,
are novel in the context of existing research on the initial segment complexity of universal probabilities, and depends significantly on a number of non-trivial results from algorithmic randomness.
We also demonstrated in the last section that these ideas are applicable to related problems
regarding universal probabilities.


\end{document}